\documentclass[orivec,12pt]{llncs}
\usepackage{lmodern}
\usepackage[utf8]{inputenc}
\usepackage[T1]{fontenc}

\usepackage{amsmath}
\usepackage{amssymb}
\usepackage{latexsym}
\usepackage{graphicx}
\usepackage{algorithm}
\usepackage{algorithmic}

\usepackage{afterpage}

\let\oldcaption\caption
\renewcommand{\caption}[1]{\oldcaption{{\fontsize{10}{12}\selectfont #1}}}

\newcommand{\hide}[1]{}
\newcommand{\stack}{\mbox{\sf StackMST}}
\newcommand{\np}{\mbox{\sf NP}}
\newcommand{\apx}{\mbox{\sf APX}}
\newcommand{\p}{\mbox{\sf P}}

\newcommand{\fptas}{\mbox{\sf FPTAS}}

\newcommand{\comp}{\textup{\texttt{comp}}}
\newcommand{\cycle}{\textup{\texttt{cycle}}}

\title{Specializations and Generalizations of the \\Stackelberg Minimum Spanning Tree Game\thanks{A preliminary version of this paper was published in the \emph{Proceedings of the 6th International Workshop on Internet and Network Economics (WINE 2010)}, and appeared in Vol. 6484 of Lecture Notes in Computer Science, Springer, 75--86. DOI: http://dx.doi.org/10.1007/978-3-642-17572-5\_7 -- This work is partially supported by the Research Grant PRIN 2010 ``ARS    TechnoMedia'' (Algorithms for Techno-Mediated Social Networks), funded by the Italian Ministry of Education, University, and Research.}}

\author{Davide Bilò\inst{1} \and Luciano Gualà\inst{2} \and Stefano Leucci\inst{3} \and Guido Proietti\inst{3,4}}

\institute{
Dipartimento di Scienze Umanistiche e Sociali, University of Sassari, Italy. \and
Dipartimento di Ingegneria dell'Impresa, University of Rome ``Tor Vergata'', Italy. \and
Dipartimento di Ingegneria e Scienze dell'Informazione e Matematica, University of L'Aquila, Italy. \and
Istituto di Analisi dei Sistemi ed Informatica, CNR, Rome, Italy. \\
\email{davidebilo@uniss.it, guala@mat.uniroma2.it, stefano.leucci@graduate.univaq.it, guido.proietti@univaq.it}
}

\date{}

\begin{document}
\maketitle

\begin{abstract}
{\fontsize{10}{12} \selectfont
Let be given a graph $G=(V,E)$ whose edge set is partitioned
into a set $R$ of \emph{red} edges and a set $B$ of
\emph{blue} edges, and assume that red edges are weighted and form a spanning tree of $G$.
Then, the \emph{Stackelberg Minimum Spanning Tree} (\stack) problem is that of pricing (i.e., weighting) the blue edges in such a way that the total weight of the blue edges selected in a minimum spanning tree of the resulting graph is maximized.
\stack \ is known to be \apx-hard already when the number of
distinct red weights is 2. In this paper we analyze some meaningful
specializations and generalizations of \stack, which shed some
more light on the computational complexity of the problem. More
precisely, we first show that if $G$ is restricted to be \emph{complete}, then the
following holds: (i) if there are only 2 distinct red weights, then
the problem can be solved optimally (this contrasts with the
corresponding \apx-hardness of the general problem); (ii)
otherwise, the problem can be approximated within $7/4 +
\epsilon$, for any $\epsilon > 0$. Afterwards, we define a natural
extension of \stack, namely that in which blue edges have a
non-negative \emph{activation cost} associated, and
it is given a global \emph{activation budget} that must not be exceeded when pricing blue edges. Here,
after showing that the very same approximation ratio as that of
the original problem can be achieved, we prove that if the spanning tree of red edges can be rooted so as that any root-leaf path contains at most $h$ edges,
then the problem admits a $(2h+\epsilon)$-approximation algorithm, for any $\epsilon > 0$.

\bigskip

}
\end{abstract}

\setcounter{footnote}{0}

\section{Introduction}
Leader-follower games, which were introduced by von Stackelberg in the far 1934~\cite{Sta34}, have recently received a considerable attention from the computer science community. This is mainly due to the fact that the Internet is a vast electronic market composed of millions of independent end-users (i.e., the followers), whose actions are by the way influenced by a limited number of owners of physical/logical portions of the network (i.e., the leaders), that can set the price for using their own network links. In particular, in a scenario in which the leaders know in advance that the followers will allocate a communication subnetwork enjoying some criteria, a natural arising problem is that of analyzing how the leaders can optimize their pricing strategy. Games of this latter type are widely known as \emph{Stackelberg Network Pricing Games} (SNPGs).

When only 2 players (i.e., a leader and a follower) are involved, a SNPG can be formalized as follows: We are given a graph $G=(V,E)$, whose edge set is partitioned
into a set $R$ of \emph{red} edges and a set $B$ of
\emph{blue} edges, and an edge cost function $c:R \rightarrow \mathbb{R}^+$ for red edges only, while blue edges need
instead to be priced by the leader. In the following, we assume that $n=|V|$ and $m=|R|+|B|$. Then, the leader moves first and chooses a
pricing function $p :B \rightarrow \mathbb{R}^+$ for her\footnote{Throughout the paper, we adopt the
convention of referring to the leader and to the follower with female and male pronouns, respectively.} edges, in an attempt to \emph{maximize} her
objective function $f_1(p,H(p))$, where $H(p)$ denotes the decision which will be taken by the follower,
consisting in the choice of a subgraph of $G$. This notation stresses the fact that the leader's problem is
implicit in the follower's decision. Once observed the leader's choice, the follower reacts by selecting a
subgraph $H(p)=(V',E')$ of $G$ which \emph{minimizes} his objective function $f_2(p,H)$, parameterized in $p$.
Note that the leader's strategy affects both the follower's objective function and the set of feasible
decisions, while the follower's choice only affects the leader's objective function. Quite
naturally, we assume that $f_1$ is \emph{price-additive}, i.e., $f_1(p,H(p))=\sum_{e \in B \cap E'} p(e)$. This
means, the leader decides edge prices having in mind that her revenue equals the overall price of her selected
edges. Therefore, the 2-player game can be equivalently thought (as we will do in the rest of the paper) as a \emph{bilevel} optimization problem in which an optimal value of $f_1$ has to be computed.

\paragraph{\it Previous work.} The most immediate SNPG is that in which we are given two
specified nodes in $G$, say $s,t$, and the follower wants to
travel along a \emph{shortest path} in $G$ between $s$ and $t$
(see \cite{VH06} for a survey). This problem has been shown to be
\apx-hard~\cite{Jor09}, as well as not approximable within a factor of $2-o(1)$ unless $\p=\np$~\cite{BK09}, while an
$O(\log |B|)$-approximation algorithm is provided in~\cite{RSM05}. For the case of multiple followers (each with a
specific source-destination pair), Labbé \emph{et al.}
\cite{LMS01} derived a bilevel LP formulation of the problem (and
proved \np-hardness), while Grigoriev \emph{et al.} \cite{Gri05}
presented algorithms for a restricted shortest path problem on
parallel edges. Furthermore, when all the followers share the same
source node, and each node in $G$ is a destination of a single
follower, then the problem is known as the \emph{Stackelberg
single-source shortest paths tree} game. In this game, the
leader's revenue for each selected edge is given by its price
multiplied by the number of paths -- emanating from the source --
it  belongs to, and in \cite{BGPW08} it was proved that finding an
optimal pricing for the leader's edges is \np-hard, as soon as
$|B|=\Theta(n)$.

Another basic SNPG, which is of interest for this paper, is that
in which the follower wants to use a \emph{minimum spanning tree}
(MST) of $G$. For this game, known
as \emph{Stackelberg MST} (\stack) game, in \cite{Car07} the
authors proved the \apx-hardness already when the number of red
edge costs is 2, and gave a $\min\{k, 1 + \ln \beta, 1 + \ln
\rho\}$-approximation algorithm, where $k$ is the number of
distinct red costs, $\beta$ is the number of blue edges selected
by the follower in an optimal pricing, and $\rho$ is the maximum
ratio between red costs. In a further paper \cite{Car13}, the
authors proved that the problem remains \np-hard even if $G$ is
planar, while it can be solved in polynomial time once that $G$
has bounded treewidth. We point out that a structural property about \stack, which will also hold for our generalized version we are going to present, is that the hardness in finding an optimal solution lies in the selection of the optimal set of blue edges that will be purchased by the follower, since once that a set of blue edges is part of the final MST, then their best possible pricing can be computed in polynomial time, as shown in~\cite{Car07}.

Notice that all the above examples fall within the class of SNPGs
handled by the general model proposed in \cite{BHK12},
encompassing all the cases where each follower aims at optimizing a
polynomial-time network optimization problem in which the cost of
the network is given by the sum of prices and costs of contained
edges. Nevertheless, SNPGs for models other than this one have been studied in~\cite{BGP09,Briestetal12}.

\paragraph{\it Our results.} In this paper we analyze some meaningful specializations
 and generalizations of \stack, which shed some more light on the computational complexity of the game.
 For the sake of presenting our results in a unifying framework, we start by defining the
 aforementioned generalized version of \stack. First of all, notice that given any instance of \stack, this
 can be simplified into an equivalent instance in which we compute a red MST of $G$, and then we discard all the
 red edges not belonging to it (see also \cite{Car07}). Then, the \emph{budgeted} \stack\ game is a 2-player game defined as follows.
 We are given a tree $T=(V,E(T))$ of $n$ nodes where each (red) edge $e \in E(T)$ has a fixed non-negative cost $c(e)$.
Moreover, we are given a non-negative \emph{activation cost} $\gamma(e)$ for
each (blue) edge $e=(u,v) \notin E(T)$, and a budget $\Delta$. The game,
denoted by \stack$(\gamma,\Delta)$, consists of two phases. In the
first phase the leader selects a set $F$ of edges to add to
$T$ such that the budget is not exceeded, i.e., $\sum_{e \in F}
\gamma(e) \leq \Delta$, and then prices them with a price function
$p:F \rightarrow \mathbb{R}^+$ having in ming that, in the second phase, the
follower will take the weighted graph $G=(V,E(T) \cup F)$
resulting from the first phase, and will compute a MST $M(F,p)$ of
$G$. Then, the
leader will collect a revenue of $r(M(F,p))= \sum_{e \in F \cap
M(F,p)} p(e)$. Our goal is to find a strategy for the leader which
maximizes her revenue.\footnote{Throughout the paper, as usual we assume that when multiple
optimal solutions are available for the follower, then he selects
an optimal solution maximizing the leader's revenue.}
Notice that using this more general definition, the original \stack\ game can be rephrased as a \stack$(\gamma,\Delta)$ game in which $T$ is any red MST of $G$, $\Delta$ is equal to 0, and the activation cost for an edge not in $E(T)$ is equal to 0 if it belongs to $B$, otherwise it is equal to any positive value.

In this paper, we prove the following results:

\begin{enumerate}
\item \stack$(0,0)$ with only 2 distinct red costs can be solved
optimally, where the first 0 in the argument is used to denote the
fact that $\gamma$ is identically equal to 0; in other words, this is a special case of \stack\  with only two red edge costs in which the input graph is complete;

\item \stack$(0,0)$ can be approximated within $3/2 + \epsilon$,
for any $\epsilon > 0$ when the red edges form a path;

\item \stack$(0,0)$ can be approximated within $7/4 + \epsilon$,
for any $\epsilon > 0$, in general;

\item \stack$(\gamma,\Delta)$ admits a $\min\{k, 1 + \ln \beta, 1
+ \ln \rho, 2h+\epsilon\}$-approximation algorithm, for any
$\epsilon>0$, where $k$, $\beta$ and $\rho$ are as previously
defined for \stack, and $h$ denotes the \emph{radius} of $T$ w.r.t. the number of edges, once $T$ is rooted at its center.
\end{enumerate}

We point out that all the above problems have an
application counterpart, since the \stack$(0,0)$ class of problems
models the case in which the leader retains the potentiality to
activate (at no cost) any missing connection in the network, while clearly
result (4) complements the approximation ratio given in~\cite{Car07} whenever the radius of the red tree is bounded, which
might well happen in practice. Finally, notice also that
\stack$(0,0)$ is a specialization of the general \stack, for which
however we were not able to prove whether the problem is in $\p$ or not.
Therefore, this remains a challenging open problem.

The rest of the paper is organized by providing each of the above
results in a corresponding section, followed by a concluding section listing some interesting problems left open.


\section{Exact algorithm for \stack$(0,0)$ with  costs in $\{a,b\}$}
In this section we present an exact polynomial-time algorithm for
\stack$(0,0)$ when the cost of any red edge belongs to the set $\{a,b\}$, with $0 \le a < b$.
Notice that this case is
already \apx-hard for \stack{} \cite{Car07}. For the sake of clarity,
we will first present the algorithm and the analysis when the red tree
is actually a path. The extension to the general case will be derived in the
subsequent subsection.

\subsection{Solving \stack$(0,0)$ with two red costs when $T$ is a path}

Now, we present an exact algorithm for
\stack$(0,0)$ on a red path $P$ with costs in $\{a,b\}$, with $0 \le a < b$.
We call a subpath $P'$ of $P$ an \emph{$a$-block} if
$P'$ has all edges of cost $a$, and $P'$ is maximal (w.r.t.
inclusion). We say that an $a$-block is \emph{good} if its length
is greater than or equal to 3, \emph{bad} otherwise. Let $\sigma$ be
the number of bad blocks of $P$.

We first present an algorithm achieving a revenue of $c(P)- \min \Big\{\sigma a,  \left\lfloor
\frac{\sigma}{2} \right\rfloor (b-a) +  \left(\sigma-2 \left\lfloor
\frac{\sigma}{2}\right \rfloor \right) \min\{a,b-a\} \Big\}$, where $c(P)$ denotes the sum of costs of edges of $P$, and then we show that such a revenue is actually an upper bound to the optimal revenue.

For technical convenience, we only consider instances where $P$ has at least $5$ edges. Clearly, the solutions for the remaining instances can be easily computed. The algorithm uses the following four rules. Each rule considers a subpath of $P$ and specifies a feasible solution
for the subpath, i.e., a set of blue edges incident to the vertices of the subpath with
a corresponding pricing. The solutions corresponding to the rules
are shown in Figure \ref{fig:a-b-path_rules}.

\begin{description}
\item[Rule 1:] Let $P'$ be a subpath of $P$ containing only one $a$-block,
and this $a$-block is good. We can obtain revenue $c(P')$ from
$P'$ by adding blue edges only within $P'$.

\item[Rule 2:] Let $P'$ be a subpath of $P$ containing only one $a$-block
and this $a$-block is bad. We can compute a solution with revenue
$c(P')-a$ from $P'$.

\item[Rule 3:] Let $P'$ be a subpath of $P$ containing one $a$-block,
and this $a$-block is the last bad block of $P$. Moreover $P'$ has at least one more edge of cost $b$ that either precedes or follows the $a$-block.
We can obtain a revenue of $c(P') - (b-a)$ from $P'$ by using a star of blue edges centered at the left or right endvertex of $P$, depending on the position of the edge of cost $b$. Notice that the endvertices of $P'$ might be followed by other good blocks.

\item[Rule 4:] Let $P_1, P_2$ be two edge-disjoint subpaths of $P$ each
containing only one $a$-block. Assume that both $a$-blocks are bad
and $P_1$ contains an edge of cost $b$ whose removal separates the
two $a$-blocks. We can obtain a revenue of $c(P_1)+c(P_2) - (b-a)$
from $P_1$ and $P_2$. Notice that $P_1$ and $P_2$ do not need to be adjacent.
\end{description}

\begin{figure}[t]
\begin{center}
\includegraphics{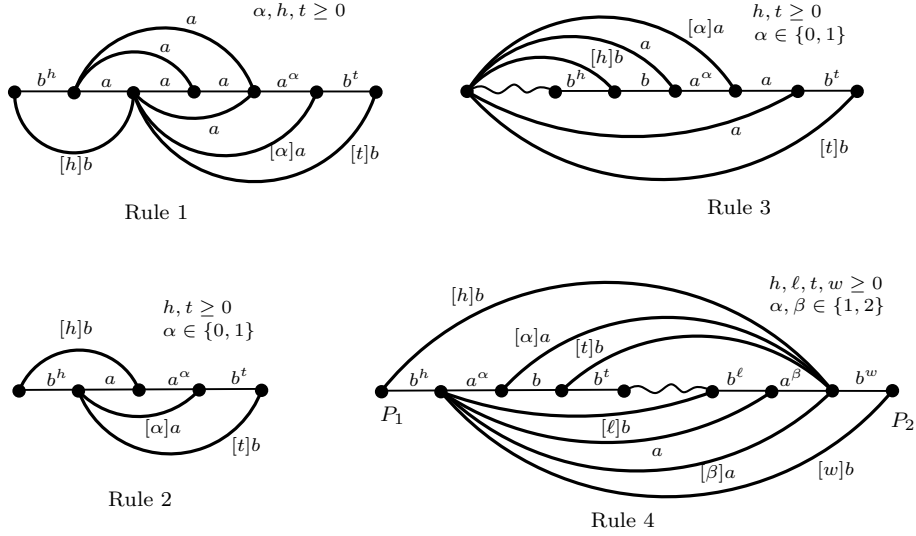}
\end{center}
\caption{Rules used by the algorithm to solve subpaths. We denote
by $\eta^\delta$ a path of $\delta$ edges each having a cost of $\eta$. An
edge with label $[i]\eta$ represents $i$ blue edges each having a price of
$\eta$. Observe that, except for Rule 2, all red (path) edges with
cost $a$ will be discarded by the follower. Concerning Rule 2, the follower
will select only a single red edge of cost $a$ (of the shown subpath). The left endvertex of the path of Rule 3 corresponds to one endvertex of $P$.}
\label{fig:a-b-path_rules}
\end{figure}

Our algorithm is as follows. If $b \ge 3a$ then we split $P$ into
subpaths each of them containing exactly one $a$-block. Then we
apply Rule 1 or Rule 2 to each subpath, depending on whether the $a$-block
in the subpath is good or bad. Hence, this solution yields a revenue of
$c(P)-\sigma a$.

Now, consider the case $b<3a$. Let $B_1, \dots, B_{\sigma}$ be the bad
$a$-blocks contained in $P$ from left to right w.r.t one of the endvertices.
We first consider the case where $\sigma \ge 2$, i.e., there are at least two bad blocks.
The algorithm splits $P$ into subpaths such that (i) each subpath contains
exactly one $a$-block, (ii) for every $i=0,\dots, \lfloor \sigma/2
\rfloor -1$, subpath containing $B_{2i+1}$ has an edge of cost $b$
incident to its right endvertex, and (iii) if $\sigma$ is odd, the
subpath containing $B_\sigma$ has an edge of cost $b$ incident to its
left endvertex.\footnote{Property (iii) can always be guaranteed since $\sigma \ge 2$.}
Let $P_i$ be the subpath containing $B_i$. The algorithm uses Rule 4 for every pairs of subpaths $P_{2i+1},P_{2i+2}$, $i=0,\dots,
\lfloor \sigma/2 \rfloor -1$, and Rule 1 for every subpath containing a good
$a$-block. Finally, if $\sigma$ is odd, we apply Rule 3 for $P_\sigma$ when
$b \le 2a$ (we can apply Rule 3 since property (iii) above holds), while we use Rule 2 when $b>2a$. It is easy to see
that the revenue of this solution coincides with $c(P)- \min \Big\{\sigma a,  \left\lfloor
\frac{\sigma}{2} \right\rfloor (b-a) + \left(\sigma-2 \left\lfloor
\frac{\sigma}{2}\right \rfloor \right) \min\{a,b-a\} \Big\}$.

Concerning the case $\sigma \le 1$, then either $P$ has no bad blocks, and then a revenue of $c(P)$ can be obtained, or there exists only one bad block $B_1$.
In this latter case:

\begin{itemize}
\item if $b \ge 2a$, let $P'$ be any subpath containing $B_1$; then, solve $P'$ using Rule 2.
\item Otherwise, if $b < 2a$, let $P'$ be a subpath containing $B_1$ and a suitable additional edge of cost $b$; then, use Rule 3 on $P'$.
\end{itemize}

\noindent
Finally, split $P \setminus P'$ into subpaths, each contaning one good $a$-block, and solve them using Rule 1.
By doing so we obtain a revenue of $c(P) - \min\{a, b-a\}$.

Now, we show that the revenue computed by the above algorithm is the optimal revenue $r^*$:

\begin{lemma}\label{thm:UB for a-b path}
$r^* \le c(P)- \min \left\{\sigma a,  \left\lfloor
\frac{\sigma}{2} \right\rfloor (b-a) +  \left(\sigma-2 \left\lfloor
\frac{\sigma}{2}\right \rfloor \right) \min\{a,b-a\} \right \}$.
\end{lemma}
\begin{proof}
Let $n_a$ be the
number of red edges of cost $a$. Let $T^*$ be the tree computed by
the follower w.r.t. an optimal solution. Moreover, let $B_1,
\dots, B_\sigma$ and $\hat{B}_1,\dots,\hat{B}_{\sigma'}$ be the bad and the
good blocks of $P$, respectively. We denote by $m_i$ and
$\hat{m}_j$ the number of edges of $B_i$ and $\hat{B}_j$,
respectively. Moreover, for an edge $e=(x,y)$, $T^*(e)$ will
denote the unique path in $T^*$ between $x$ and $y$ (observe that $T^*(e)$ may be the path containing only edge $e$). For each
$i=1,\dots,\sigma$ and $j=1,\dots,\sigma'$, consider the \emph{bad tree} $T_i=\bigcup_{e \in
E(B_i)} T^*(e)$,\footnote{Here the union symbol denotes the union of graphs.} and the \emph{good tree} $\hat{T}_j=\bigcup_{e \in E(\hat{B}_j)}
T^*(e)$.
Let ${\cal T}=\{T_1,\dots,T_\sigma\} \cup\{\hat T_1,\dots, \hat T_{\sigma'}\}$. 
Observe that for each $i,j$, we have: (i) $T_i$ and $\hat{T}_j$ are trees and
every edge has cost $a$, (ii) $V(B_i)\subseteq V(T_i)$ and
$V(\hat{B}_j)\subseteq V(\hat{T}_j)$, and (iii) $E(T^*) \cap
E(B_i)\neq \emptyset$, or $T_i$ contains at least $m_i+1$ edges.

Let us consider the following graph $H=(\bigcup_i V(T_i) \cup
\bigcup_j V(\hat{T}_j), \bigcup_i E(T_i) \cup \bigcup_j
E(\hat{T}_j))$, and let $N$ be the number of nodes of $H$. Clearly,
$H$ is a forest. Moreover, each connected component of $H$ is either a single tree of ${\cal T}$ or it consists of the union of at least two trees in ${\cal T}$. Let us consider the set $X$ of ``unmerged''  bad trees, i.e., $X=\{T_i \mid i=1,\dots,\sigma, V(T_i)\cap V(T)=\emptyset, \,\, \forall \,\, T \in {\cal T}\setminus\{T_i\} \}$. We define $\ell = |X|$. Observe that each tree in $X$ is in the set $\cal C$ of connected components of $H$. Let $t$ be the number of the remaining connected components of $H$, i.e., $|{\cal C}|=t+\ell$. As each bad tree not in $X$ has been merged with some other tree, we have $t \le \sigma' + \left \lfloor\frac{\sigma-\ell}{2}
\right \rfloor$.

In order to relate $t$ to the number $N$ of nodes of $H$, we define
$\ell_1=|\{ T_i \mid T_i \in X, E(T^*)\cap E(B_i)\neq \emptyset
\}|$. Notice that $\ell_1$ is a lower bound to the number of red
edges in $H$. We now give a lower bound to $N$. Since $H$ spans all $a$-blocks (which are pairwise vertex disjoint), and since property (iii) holds, we have that $N \ge n_a + \sigma + \sigma' + \ell - \ell_1$. Therefore, since $H$ has $N- \ell -t$ edges of cost $a$, and using $c(P)= n_a (a-b) + (n-1)b$, we have:
\begin{eqnarray*}
r^* &\le& \big(N- \ell - t\big)a - \ell_1 a + \Big(n-1-\big(N-\ell-t\big)\Big)b \\
    &= & (N-\ell-t)(a-b) - \ell_1 a + (n-1)b \\
    & \leq & (n_a+\sigma+\sigma'-\ell_1-t)(a-b) +(n-1)b-\ell_1 a\\
    &=& c(P)-\big(\sigma+\sigma'-\ell_1-t\big)(b-a)-\ell_1 a\\
    &\le& c(P)-\left((\sigma-\ell_1)-\left\lfloor
    \frac{\sigma-\ell}{2}\right\rfloor\right)(b-a)-\ell_1 a\\
    &\le& c(P)- \min \left\{\sigma a,  \left\lfloor
\frac{\sigma}{2} \right\rfloor (b-a) +  \left(\sigma-2 \left\lfloor
\frac{\sigma}{2}\right \rfloor \right) \min\{a,b-a\} \right \}.
\end{eqnarray*}
To see why the latter inequality holds, one can consider the different parity of $\sigma$ and $\ell$ for each of the following three cases: $b \ge 3a$, $2a \le b < 3a$, and $b<2a$. \qed
\end{proof}

Hence, from the above lemma, we
have:
\begin{theorem}
\stack$(0,0)$ can be solved in polynomial time when the red edges form a path and their costs are in $\{a,b\}$.
\end{theorem}


\subsection{Solving \stack$(0,0)$ with two red costs: the general case}

We now extend the previous result by providing an optimal polynomial-time algorithm for
\stack$(0,0)$ when the red costs belong to the set $\{a,b\}$, and $T$ is a tree.

Let $0 \le a < b$, and let $T$ be a red tree with costs in
$\{a,b\}$. In a similar fashion as before, we call a subtree
$T'$ of $T$ an \emph{$a$-block} if
$T'$ has all edges of cost $a$, and $T'$ is maximal (w.r.t.
inclusion). We say that an $a$-block is \emph{bad} if it is a
star, \emph{good} otherwise. Let $\sigma$ be the number of bad blocks
of $T$. As the upper bound to the maximum revenue
$r^*$ shown in Lemma~\ref{thm:UB for a-b path} still holds,\footnote{The proof is identical to the one previously shown.} we now present a general algorithm achieving a revenue equal to the given upper bound.

The four rules used by the algorithm are similar to
the ones used in the algorithm for the path and they are shown in
Figure \ref{fig:a-b-tree_rules_1} and \ref{fig:a-b-tree_rules_2}, along with the corresponding revenues.

Our algorithm is as follows. If $b \ge 3a$, then we split $T$ into
subtrees, each of them containing exactly one $a$-block. Then we
apply Rule 1 or Rule 2 to each subtree, depending on whether the $a$-block
in the subtree is good or bad. Clearly, this solution yields a revenue of
$c(T)-\sigma a$

\begin{figure}[t]
\begin{center}
	\includegraphics[scale=0.9]{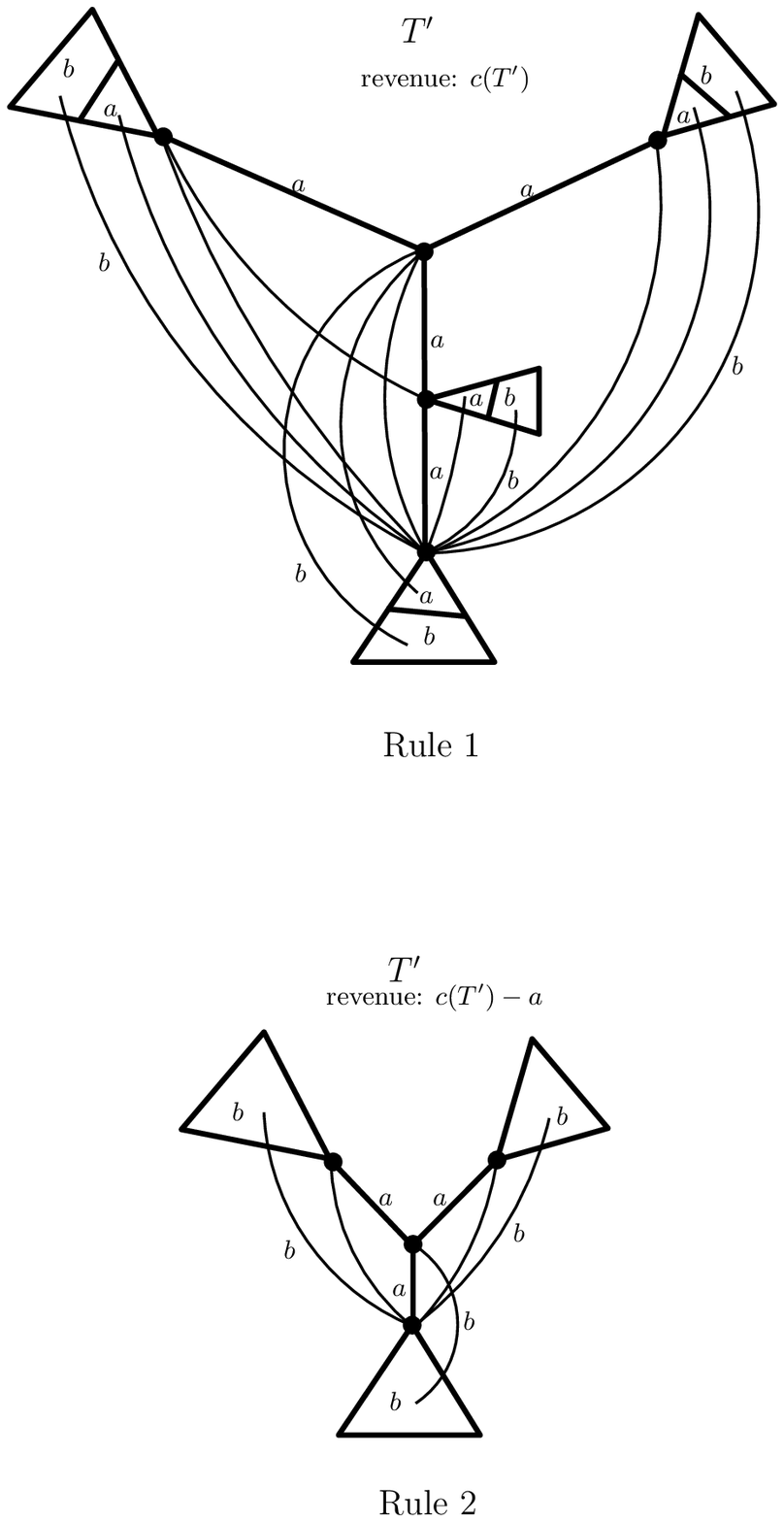}
\end{center}\caption{Rules 1 and 2 used by the algorithm to solve subtrees.
Edges without label are priced to $b$.}\label{fig:a-b-tree_rules_1}
\end{figure}

\begin{figure}[t]
\begin{center}
	\includegraphics[scale=0.85]{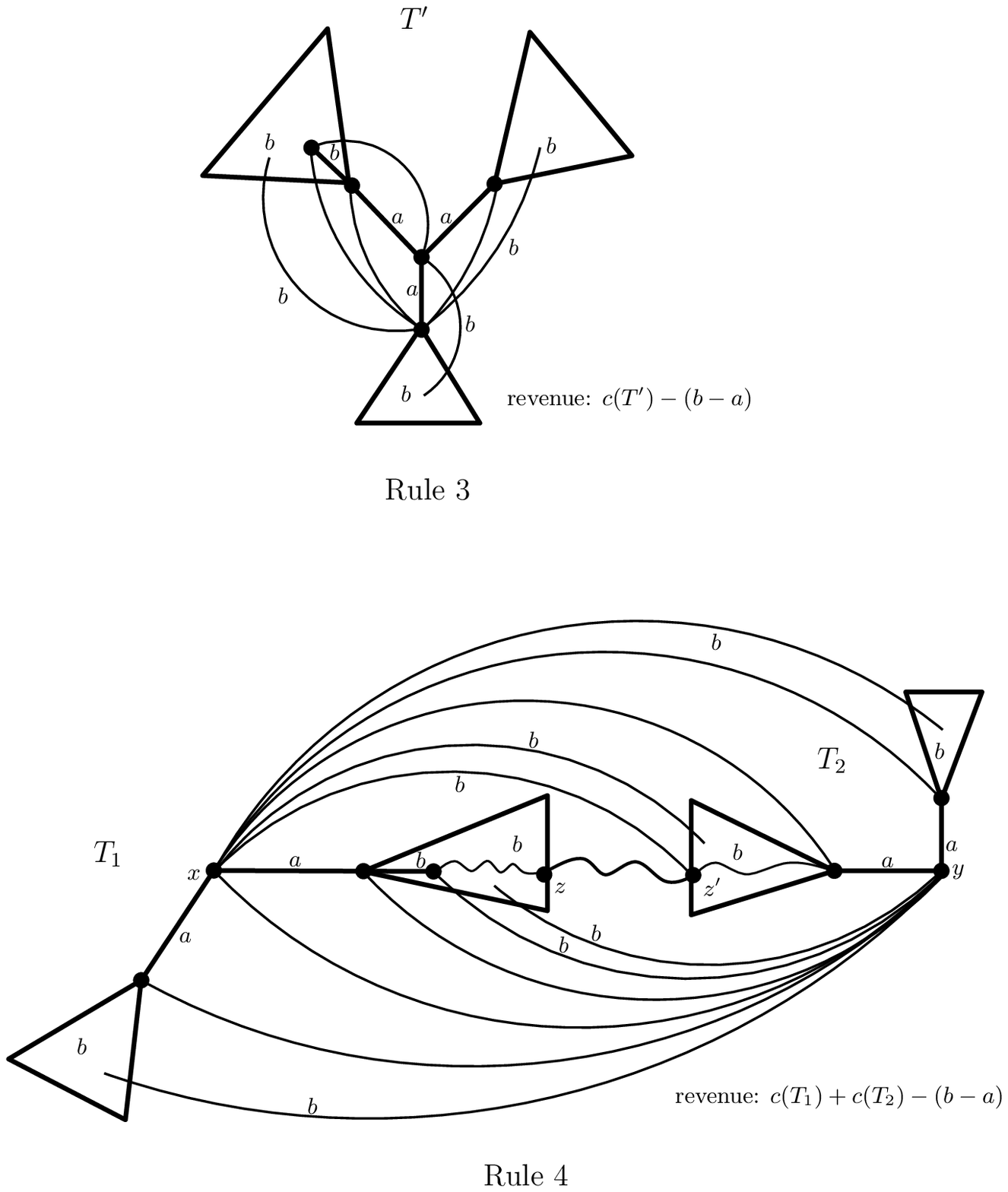}
\end{center}
\caption{Rules 3 and 4 used by the algorithm to solve subtrees.
Edges without label are priced to $b$. Notice that in Rule 4,
while there is a blue edge between $x$ and $z'$, there is no edge
between $y$ and $z$.} \label{fig:a-b-tree_rules_2}
\end{figure}

\afterpage{\clearpage}

Now, consider the case $b<3a$. Let $B_1, \dots, B_\sigma$ be the bad
$a$-blocks contained in $T$. The algorithm splits $T$ into
subtrees such that (i) each subtree contains exactly one
$a$-block, (ii) there exists a permutation $B'_1, \dots, B'_\sigma$ of
the bad $a$-blocks such that for every $i=0,\dots, \lfloor \sigma/2
\rfloor -1$, subtree containing $B'_{2i+1}$ has an edge of cost $b$
along the (unique) path joining $B'_{2i+1}$ with $B'_{2i+2}$, and
(iii) if $\sigma$ is odd, the subtree containing $B'_\sigma$ has an edge of
cost $b$.

Let $T_i$ be the subtree containing $B'_i$. The algorithm uses
Rule 4 for every pair of subtrees $T_{2i+1},T_{2i+2}$,
$i=0,\dots, \lfloor \sigma/2 \rfloor$, Rule 1 for every subtree
containing a good $a$-block. Finally, if $\sigma$ is odd, we apply Rule 3 for $T_\sigma$ when $b \le 2a$, while we use Rule 2 when $b>2a$.
From this, we have:

\begin{theorem}
\stack$(0,0)$ can be solved in polynomial time when red
edge costs are in $\{a,b\}$.
\end{theorem}


\section{\stack$(0,0)$ can be approximated within $3/2 + \epsilon$ when the red edges form a path}

Here we design a $(\frac{3}{2} + \epsilon)$-approximation
algorithm for \stack$(0,0)$ when the tree $T$ is actually a path, say $P$. 

Let then $P$ be the path of red edges. The idea of the algorithm is to consider
three possible solutions and pick the best one. We will argue that the
revenue of such a solution is at least a fraction $\frac{2}{3}$ of
the cost of almost the entire path. More precisely, we select a
cheap subpath $\bar{P}$ of $P$ of length 2 or 3, and we then compute a
solution achieving a revenue of at least $\frac{2}{3}
(c(P)-c(\bar{P}))$.

Let $m=n-1$ be the length of $P$, and let $e_1,\dots,e_{m}$ be the
red edges of $P$ in the order of traversing $P$ from an endpoint
to the other one. Moreover, let us set $\ell$ to 2, if $m$ is
even, and to 3 otherwise. Let $P_i$ be the subpath of $P$ of
length $\ell$ starting from $e_i$, i.e., $P_i$ consists of the
edges $e_i,\dots,e_{i+\ell-1}$. Let $\bar{P}$ be the subpath with
minimum cost among $P_{2j-1}$, $j=1,\dots, \lfloor m/2 \rfloor$.
If we remove $\bar{P}$ from $P$, we obtain two paths of even
length, say $Q_1$ and $Q_2$. At most one of $Q_1$ and $Q_2$ may be
empty. Let us assume for the ease of presentation that both paths are
non-empty (similar arguments hold when this does not happen), and
let $2h$ and $2k$ be the length of $Q_1$ and $Q_2$, respectively.
Moreover, let $u_0,u_1,\dots,u_{2h}$ and
$v_0,v_1,\dots,v_{2k}$ be the nodes of $Q_1$ and $Q_2$,
respectively. The two end-nodes of $\bar{P}$ are $u_{2h}$ and
$v_0$. Let $x_i=c(u_{i-1},u_i)$ and $y_j=c(v_{j-1},v_j)$.
Finally, let $z$ be an internal node of $\bar{P}$ (see Figure
\ref{fig:path}).

\begin{figure}[t]
\begin{center}
\includegraphics[scale=0.9]{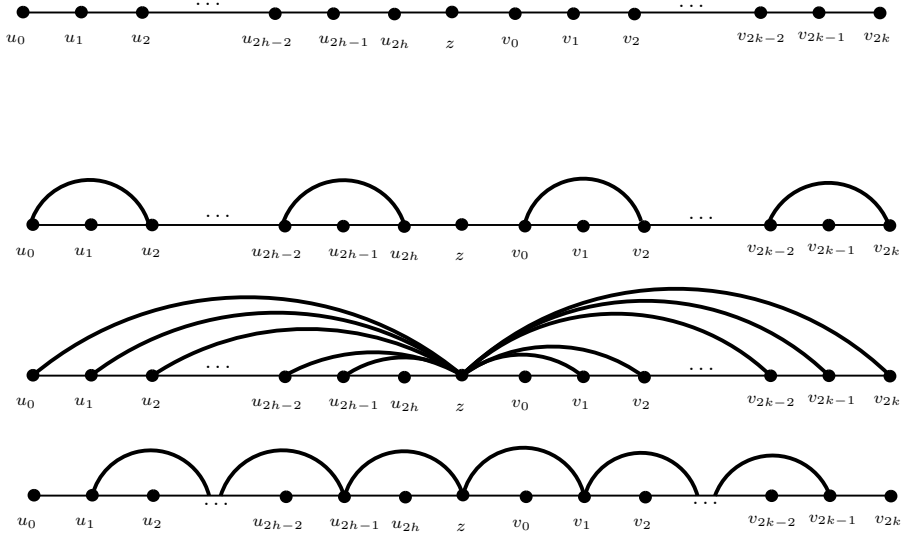}
\end{center}
\caption{The path and the three solutions considered by the
algorithm. Blue edges are in bold. Here $\bar{P}$ consists of 2
edges.} \label{fig:path}
\end{figure}

Let $A=\sum_{i=1}^h \max\{x_{2i-1},x_{2i}\} + \sum_{i=1}^k
\max\{y_{2i-1},y_{2i}\}$, and let $B= \sum_{i=1}^h
\min\{x_{2i-1},x_{2i}\} + \sum_{i=1}^k \min\{y_{2i-1},y_{2i}\}$.
Notice that $c(Q_1)+c(Q_2)=A+B$. The first solution we consider is
\begin{equation*}
F_1=\{(u_{2i-2},u_{2i}) \mid i=1,\dots,h \} \cup
\{(v_{2i-2},v_{2i}) \mid i=1,\dots,k \},
\end{equation*}

\noindent and the price function is defined as
$p(u_{2i-2},u_{2i})=\max\{x_{2i-1},x_{2i}\}$, and
$p(v_{2i-2},v_{2i})= \linebreak \max\{y_{2i-1},y_{2i}\}$. Notice that this
solution obtains a revenue $r_1 = A$.

The second solution is a star centered in the node $z$; more
precisely:
\begin{equation*}
F_2=\{(z, u_{i}) \mid i=0,\dots,2h-1 \} \cup \{(z,v_{i}) \mid
i=1,\dots,2k \},
\end{equation*}

\noindent and the prices are defined as $p(z,u_0)=x_1,
p(z,v_{2k})=y_{2k}, p(z,u_i)=\min\{x_i,x_{i+1}\}$, and
$p(z,v_i)=\min\{y_i,y_{i+1}\}$. Notice that this solution obtains
a revenue of
\begin{equation*}
r_2 = B+ x_1 + y_{2k} + \sum_{i=0}^{h-2} \min\{x_{2i+2},x_{2i+3}\}
+ \sum_{i=0}^{k-2} \min\{y_{2i+2},y_{2i+3} \}.
\end{equation*}

Finally, the third solution is the following:
\begin{multline*}
F_3=\{(u_{2i+1},u_{2i+3}) \mid i=0,\dots,h-2 \} \cup \\
\{(v_{2i+1},v_{2i+3}) \mid i=0,\dots,k-2 \} \cup
\{(u_{2h-1},z),(z,v_1)\},
\end{multline*}

\noindent and the pricing is as follows:
$p(u_{2h-1},z)=x_{2h},p(z,v_1)=y_1$,
$p(u_{2i+1},u_{2i+3})=\linebreak \max\{x_{2i+2},x_{2i+3}\}$, and
$p(v_{2i+1},v_{2i+3})=\max\{y_{2i+2},y_{2i+3} \}$. Hence, the
corresponding revenue is:
$$
r_3 =x_{2h} + y_1 + \sum_{i=0}^{h-2} \max\{x_{2i+2},x_{2i+3}\} +
\sum_{i=0}^{k-2} \max\{y_{2i+2},y_{2i+3} \}.
$$

Hence, we have:
\begin{eqnarray*}
r_1+r_2+r_3 &=& A+B+x_1+x_{2h}+y_{1}+y_{2k}+ \\
& & \sum_{i=0}^{h-2}(\min\{x_{2i+2},x_{2i+3}\}+\max\{x_{2i+2},x_{2i+3}\}) + \\
& & \sum_{i=0}^{k-2}(\min\{y_{2i+2},y_{2i+3}\}+\max\{y_{2i+2},y_{2i+3}\}) \\
\\ &=& 2(c(Q_1)+c(Q_2)),
\end{eqnarray*}

\noindent from which it follows that the revenue
$r=\max\{r_1,r_2,r_3\}$ is at least $\frac{2}{3}(c(Q_1)+c(Q_2))$. Now, observe that by construction we have

$$c(\bar{P}) \le \frac{c(P)}{\left\lfloor \frac{n}{\ell} \right\rfloor} \le \frac{3}{n-2} \left(c(\bar{P}) + c(Q_1) + c(Q_2) \right),$$

\noindent
and hence $c(\bar{P}) \le \frac{3}{n-5} \left( c(Q_1) + c(Q_2) \right)$.
Denoting by $r^*$ the optimal revenue, and observing that the cost of the red tree is always an upper bound to $r^*$, we then have
$$
\frac{r^*}{r} \leq \frac{c(P)}{r}=\frac{c(Q_1)+c(Q_2)}{r}+\frac{c(\bar{P})}{r} \le
\frac{3}{2} +\frac{\frac{3}{n-5} (c(Q_1)+c(Q_2))}{\frac{2}{3} (c(Q_1)+c(Q_2))}= \frac{3}{2} +\frac{9}{2n-10}.
$$

We have proved the following:
\begin{theorem}
\stack$(0,0)$ can be approximated within a factor of
$3/2+\epsilon$, for any $\epsilon >0$, when the red edges form a path.
\end{theorem}

We point out that our algorithm is asymptotically tight with
respect to the adopted upper-bound scheme. An example is the path
in which $c(e_1)=1, c(e_2)=2$, and $c(e_i)=0$, for every $i>2$. It
is easy to see that for this path the revenue obtained by an
optimal solution is 2, while the total cost of the path is 3.


\section{\stack$(0,0)$ can be approximated within $7/4 + \epsilon$}

In this section we design an algorithm that achieves an
approximation ratio of $7/4+\epsilon$ for the general \stack$(0,0)$ game. 

The idea of the algorithm is to partition the red tree into
suitable subtrees for which we can guarantee a revenue of at least $4/7$ of the cost of each
one of them. Let $T=(V,E(T))$ be the red tree. We say that
$T_1=(V_1,E_1),\dots,T_\ell=(V_\ell,E_\ell)$ is a partition of $T$ into $\ell$
subtrees if (i) each $T_i$ is a subtree of $T$, (ii) $V=\bigcup_i
V_i, E(T)=\bigcup_i E_i$, and (iii) for each $i,j, i\neq j$, $E_i
\cap E_j =\emptyset$. 

It is easy to see that once $T$ is partitioned into subtrees as specified above, we can solve locally a \stack$(0,0)$ game for each red subtree of the partition, and then 
solve the original problem by joining together all the local solutions (by maintaining the corresponding pricing). Indeed, the union of all the trees associated with the local solutions is clearly a spanning tree of $G$. Hence, we can claim the following

\begin{lemma}\label{lm:decomposition_property}
Let $T_1,\dots,T_\ell$ be a partition of $T$ into
$\ell$ subtrees. For each $i$, let $r_i$ be the revenue returned by a local solution of $T_i$. Then, the revenue which can be obtained for $T$ is at least $\sum_{i=1}^\ell r_i$.
\end{lemma}

Moreover, we can prove the following:

\begin{lemma}\label{lm:decomposition_alg}
Let $T$ be a tree rooted at a node $s$. There always exists a
partition of $T$ into $\ell$ subtrees $T_1,\dots,T_\ell$ such that
\begin{itemize}
    \item $T_\ell$ has at most 2 edges and at least one of them is incident to
$s$;
    \item for every $1\le j \le \ell-1$, $T_j$ is either (i) a path of 3 or
    4 edges, or (ii) a star with at least 3 edges.
\end{itemize}
Moreover, this partition can be found in polynomial time.
\end{lemma}
\begin{proof}
We provide a polynomial-time algorithm that finds the partition of
the lemma. Let 
$d(v)$ denote the depth of $v$ in $T$, i.e., the number of edges of the path (in $T$)
between $s$ and $v$. We denote by $S(v)$ the set of the children
of $v$. Moreover, we use $\bar{v}$ to denote the parent of $v$. We
proceed in phases. In phase $j$, we find a subtree $T_j$ by
applying one of the rules below (we consider them in order), then
we remove $T_j$ from $T$ and we move to the next phase. We stop when no
rule can be applied. Let $L$ be the set of leaves of $T$ with
depth equal to the current height of $T$. The rules are the following (see Figure
\ref{fig:rules}):
\begin{description}
    \item[Rule 1:] if there exists a node $v \in L$ with $d(v)\ge 2$ and
    such that $v$ has at least one sibling, then $T_j$ is the star
    with edge set $\{(\bar{v},\bar{\bar{v}})\} \cup \{(\bar{v},u) \mid u \in S(\bar{v})
    \}$;
    \item[Rule 2] if there exists a node $v \in L$ with $d(v)\ge 2$ such that
    $\bar{v}$ has a sibling $u$ and $u$ is a leaf, then $T_j$ is
    the path with edge set
    $\{(v,\bar{v}),(\bar{v},\bar{\bar{v}}),(\bar{\bar{v}},u)\}$;
    \item[Rule 3:] if there exists a node $v \in L$ with $d(v)\ge 2$ such that
    $\bar{v}$ has a sibling $u$ and $u$ is not a leaf, then let
    $u'$ be the unique child of $u$ ($u'$ must be unique otherwise Rule 1 would
    apply). Then, $T_j$ is the path with edge set
    $\{(v,\bar{v}),(\bar{v},\bar{\bar{v}}),(\bar{\bar{v}},u),(u,u')\}$;
    \item[Rule 4:] if there exists a node $v \in L$ with $d(v)\ge 3$, then $T_j$ is the path with edge
    set
    $\{(v,\bar{v}),(\bar{v},\bar{\bar{v}}),(\bar{\bar{v}},\bar{\bar{\bar{v}}})\}$;
    \item[Rule 5:] if $T$ is a star with at least 3 edges, then $T_j=T$.
\end{description}
Now, assume that the last phase is phase $\ell-1$, then we set $T_\ell$
equal to the remaining tree $T$. If there is no edge left, we set
$T_\ell$ equal to the empty subtree. It is easy to see that if $T_\ell$
is non-empty, it must have at most 2 edges, and one of them must be
incident to $s$. Moreover, since each phase takes
polynomial time and each $T_j$ with $j < \ell$ contains at least one
edge, the claim follows. \qed
\end{proof}

\begin{figure}[t]
\begin{center}
	\includegraphics{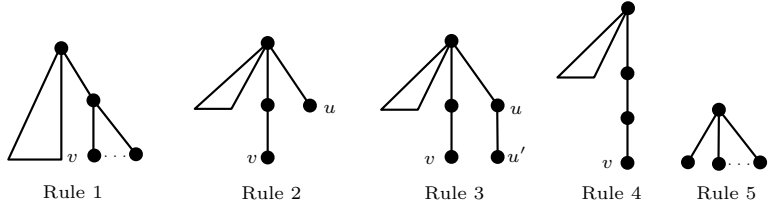}
\end{center}
\caption{The five rules of the decomposition algorithm.}
\label{fig:rules}
\end{figure}

The following lemmas allow us to obtain a revenue of at least $\frac{4}{7} \,c(T_i)$ for each subtree $T_i, i=1,\dots,\ell-1$, of the decomposition.

\begin{lemma}\label{lm:star}
Let $S$ be a star with at least 3 edges, then we can obtain a
revenue of at least $\frac{2}{3} c(S)$.
\end{lemma}
\begin{proof}
Let $s$ be the center of the star, and let $u_1,\dots,u_{t}$ be
the leaves ordered such that $c(s,u_1) \le c(s,u_2)\le \dots \le
c(s,u_{t})$. The set of blue edges $F=\{(u_1,u_j) \mid
j=2,\dots,t\}$ yields a revenue of $\sum_{j=2}^{t} c(s,u_j) \ge
\frac{2}{3} c(S)$, since $t\geq 3$. \qed
\end{proof}

\begin{lemma}\label{lm:path}
Let $P$ be a path of 3 or 4 edges, then we can obtain a revenue of
at least $\frac{4}{7} c(P)$.
\end{lemma}
\begin{proof}
Let us consider the path of 3 edges first. Let $0 \le c_1 \le c_2
\le c_3$ be the edge costs. If the cost of the middle edge is $c_1$, we can
easily obtain a revenue of $c_2+c_3 \ge \frac{2}{3}  c(P)$. Assume that
the cost of the middle edge is not $c_1$. In Figure \ref{fig:path3_4} three
solutions are shown. The corresponding revenues are: $c_3, 2 c_2,
2 c_1+ c_2$. A trivial calculation shows that the maximum of the
three revenues is at least $\frac{4}{7}  c(P)$.

\begin{figure}[t]
\begin{center}
\includegraphics{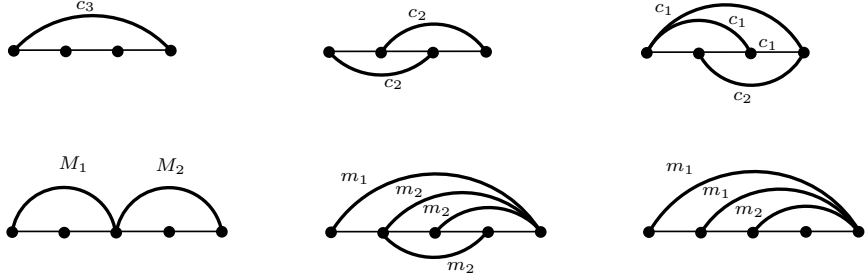}
\end{center}
\caption{Three possibile solutions for paths of 3 or 4 edges.}
\label{fig:path3_4}
\end{figure}

Now, we consider a path of 4 edges. Let $x_1,\dots,x_4$ be the
costs of the edges from left to right. We set $M_1=\max
\{x_1,x_2\}, m_1=\min \{x_1,x_2\},M_2=\max \{x_3,x_4\},m_2=\min
\{x_3,x_4\}$. Assume w.l.o.g. that $m_1 \ge m_2$. Three solutions
are shown in Figure \ref{fig:path3_4}. The corresponding revenues
are: $M_1+M_2, m_1+ 3 m_2, 2 m_1+m_2$, the maximum of which is easy to see to be at least $\frac{4}{7} c(P)$. \qed
\end{proof}

We are now ready to prove the following:

\begin{theorem}
\stack$(0,0)$ can be approximated within a factor of
$7/4+\epsilon$, for any constant $\epsilon >0$.
\end{theorem}
\begin{proof}
W.l.o.g., we can restrict ourselves to the case $n\geq \frac{7}{2\epsilon}+1$, since otherwise to find an optimal solution we can always use an exhaustive search algorithm that tries all the possible sets of blue edges and prices them at the optimum (remember this can be done in polynomial time \cite{Car07}). For each $v$, let $\mu(v)=\max_{u \mid (u,v)\in E(T)} c(u,v)$. We
root $T$ at a node $s$ minimizing $\mu$. Then we decompose $T$ using the algorithm given in Lemma
\ref{lm:decomposition_alg}, and we solve locally each $T_j$ with
$j \le \ell$. Let $r_j$ be the corresponding obtained revenue,  
and observe that $r_\ell \geq c(T_\ell)-\min_{e \in E(T_\ell)}c(e) \ge c(T_\ell) - \mu(s)$.

As Lemma~\ref{lm:decomposition_property} together with Lemmas~\ref{lm:star} and~\ref{lm:path} implies that the total revenue $r$ is at least $\sum_{j=1}^\ell r_j \ge \frac{4}{7} \left( c(T) - \mu(s) \right)$, and since $c(T) \ge \frac{1}{2}  \sum_{v \in V} \mu(v) \ge \frac{n}{2} \mu(s)$, we obtain

$$\frac{r^*}{r}\leq \frac{c(T)}{r} \le 7/4 + \frac{7}{2n-4}\leq 7/4+\epsilon.$$ \qed
\end{proof}

\section{\stack$(\gamma,\Delta)$ on trees of bounded radius}
In this section, we study the general \stack$(\gamma,\Delta)$.
First, we will argue that for this generalized version, the very same
approximation ratio as that of the original game can be achieved, since 
the \emph{single-price algorithm} defined in \cite{Car07}
can be easily adapted to provide an approximation of $\min\{k, 1 +
\ln \beta, 1 + \ln \rho\}$ for \stack$(\gamma,\Delta)$ as well,
where $k$ is the number of distinct red costs, $\beta$ is the
number of blue edges selected by the follower in an optimal
solution, and $\rho$ is the maximum ratio between red costs.
Then we focus on the case in which $T$ is a tree of radius $h$ (measured w.r.t. the number of edges) once rooted at its center.
For this case, we show that the problem remains \apx-hard even for
constant values of $h$, as well as approximable within a factor of
$2h+\epsilon$.

Let $k$ denote the number of distinct red costs, and let
$c_1<c_2< \dots < c_k$ denote these costs. To extend the single-price algorithm, we proceed as follows.
We consider the complete graph consisting of the union of the red
tree and all the potential blue edges. For each $j$ between $1$
and $k$, we set the price of every potential blue edge to $c_j$,
and we compute a spanning tree by a slightly modification of
Kruskal's algorithm as follows. In the phase in which the
algorithm considers all the edges of cost $c_j$, we break
tightness in favor of blue edges, and among the blue edges, we
prefer those with smaller activation cost. As soon as we consider
a blue edge exceeding the budget $\Delta$, we delete that edge and
all the remaining blue edges, and we go on with Kruskal's algorithm.
The solution for a given $j$ will be the set of all picked blue
edges which will be priced to $c_j$. Then we pick $j$ such that
the corresponding revenue is maximum, and we return the
corresponding solution. It turns out that the same analysis given
in \cite{Car07} can be applied here. Hence, we have:

\begin{theorem}
The above algorithm achieves an approximation ratio of $\min\{k, 1
+ \ln \beta, 1 + \ln \rho\}$ for \stack$(\gamma,\Delta)$.
\end{theorem}

We now study \stack$(\gamma,\Delta)$ when $T$ is a tree that once rooted at its center, say $v_0$, has height/radius $h$. First, we observe that the reduction the authors in~\cite{Car07} used to prove that $\stack$ is \apx-hard already when $T$ is a path can be modified to show the following:

\begin{theorem}
$\stack$ is \apx-hard even if $T$ is a star.
\end{theorem}
\begin{proof}
We show an  approximation-preserving  reduction from $\stack$ for the case in which $T$ is a path to $\stack$ for the case in which $T$ is a star. Our reduction works only for the hard instances constructed in~\cite{Car07}.

The hard instances given in~\cite{Car07} are constructed from instances of the {\em Set Cover Problem} in the following way. Let $U=\{u_1,\dots,u_\ell\}$ be a set of objects and let $\{S_1,\dots,S_t\}$ be a set of subsets of $U$ such that $u_\ell \in S_i$, for every $i=1,\dots,t$. $T$ is a path of $\ell+t$ vertices $\{u_1,\dots,u_\ell\}\cup\{S_1,\dots,S_t\}$ with edge set $E(T)=\{(u_i,u_{i+1})\mid i=1,\dots,\ell-1\}\cup\{(S_i,S_{i+1})\mid i=1,\dots,t-1\}\cup\{(u_\ell,S_1)\}$. The fixed cost $c(e)$ of an edge $e \in E(T)$ is $2$ if $e=(S_i,S_{i+1})$ or $e=(u_\ell,S_1)$, 1 otherwise. Let $B=\big\{(u_i,S_j)\mid u_i \in S_j, i=1,\dots,\ell, j=1,\dots,t\big\}$ be the set of blue edges.

Our reduction works as follows. We take the above hard instance for $\stack$ on red paths, we add a vertex $v_0$ and we replace the red tree $T$ by a star of red edges centered at $v_0$. Let $T'$ denote the star of red edges. The fixed cost $c'(e)$ of an edge $e$ is 2 if $e=(v_0,S_i)$, 1 otherwise. First observe that for every $F\subseteq B$, $(V(T),F)$ is acyclic iff $(V(T'),F)$ is acyclic. Let $F\subseteq B$ be a set of blue edges such that $(V(T),F)$ is acyclic. The revenue yielded by $F$ in both instances of $\stack$ is the same, as the price of an
edge $(S_i,u_j)\in F$ in both instances is 2 iff $(S_i,u_j)$ is the only edge in $F$ which is incident to $S_i$. The claim follows.
\qed
\end{proof}

In the remaining of the section we will show the existence of a $(2h+\epsilon)$-approximation algorithm. 
Before starting, recall that once that a set $F$ of activated edges is part of the final MST, then the optimal pricing for each $e \in F$ can be computed in polynomial time, as observed in (see~\cite{Car07}). More precisely, this can be done by computing efficiently 
\begin{equation}\label{eq: optimal pricing for set of leader's edges}
p_{F}(e):= \min_{H \in \cycle(F,e)} \ \max_{e' \in E(H)\cap E(T)} c(e')
\end{equation}
\noindent where $\texttt{cycle}(F,e)$ is the set of (simple) cycles containing edge $e$ in the graph \linebreak $(V(T), E(T) \cup F)$.
With a little abuse of notation, in the following 
we will denote by $r(F)$ the revenue yielded by the above optimal pricing $p_{F}$.

The main idea of the algorithm is to reduce the problem instance to $h$ instances in which the red trees are stars. With a little abuse of notation, in each of the $h$ instances, the leader is sometimes allowed to activate edges which are parallel to red edges. 
We denote by $V_i=\{v_1,\dots,v_{\ell_i}\}$ the set of vertices at level $i$ in $T$, and by $E_i$ the set of edges in $T$ going from vertices in $V_i$ to their parents. Let $T_i$ be a red star obtained by identifying all red edges in $T$ but those in $E_i$. With a little abuse of notation, when edge $(u,v)$ is identified, and w.l.o.g. $u$ is the parent of $v$, we assume that the corresponding vertex is labeled with $u$. Thus, according to this assumption, we have that $T_i$ is a star centered at $v_0$ with $v_1,\dots,v_{\ell_i}$ as leaves. The cost of a red edge $(v_0,v)$ in $T_i$ is $c_i(v_0,v)=c(u,v)$, where $u$ is the parent of $v$ in $T$. Let $\hat T_0,\hat T_1,\dots, \hat T_{\ell_i}$ be the connected components in $T-E_i$. W.l.o.g., assume $v_i \in V(\hat T_i)$. Let $e_{j,q}$ be a blue edge connecting $\hat T_j$ and $\hat T_q$ with cheapest activation cost. Let ${\tt blue}_i:=\{e_{j,q}\mid j,q=0,\dots,\ell_i, j \neq q\}$. Notice that this set contains edges that the leader can activate in the original instance of the problem. We now map them to their counterpart in $T_i$, namely let $B_i:=\{\bar e_{j,q}:=(v_j,v_q)\mid e_{j,q} \in {\tt blue}_i\}$ be the set of blue edges the leader is allowed to activate in $T_i$. The activation cost of an edge $\bar e_{j,q}\in B_i$ is $\gamma_i(\bar e_{j,q}):=\gamma(e_{j,q})$. The auxiliary instance corresponding to level $i$ is shown in Figure \ref{fig:2_h_approx}.

\begin{figure}[t]
\begin{center}
\includegraphics[scale=0.9]{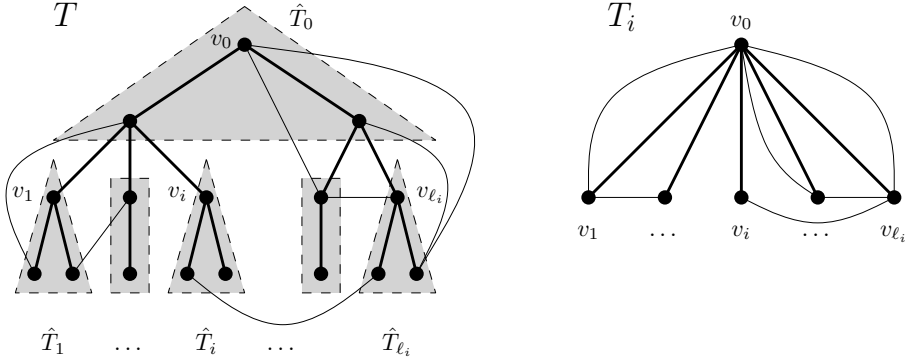}
\end{center}
\caption{Auxiliary instance for \stack$(\gamma,\Delta)$ corresponding to level $i$. Vertices in the tree $\hat{T}_q$ are identified. When two or more edges between the same two trees are present, only the one with smallest activation cost is preserved.}
\label{fig:2_h_approx}
\end{figure}

Let $F^*$ be an optimal solution for the leader on input instance $T$ and let \linebreak $F^*_i:=\{(v_j,v_q) \in B_i \mid (u,v)\in F^*, u \in V(\hat T_j),v \in V(\hat T_q), j \neq q\}$ be the corresponding edges in $T_i$. Let $G^*_i:=(\{v_0,\dots,v_{\ell_i}\},F^*_i)$, and denote by $\comp(G^*_i)$ the set of connected components of $G^*_i$. We start by proving an upper bound on the revenue yielded by $F^*$.
\begin{lemma}\label{lm: upper bound on the leader's revenue}
$\displaystyle{r(F^*)\leq c(T)-\sum_{i=1}^{h} \sum_{H \in \comp(G^*_i)} \min_{v \in V(H)}c_i(v_0,v).}$\footnote{With a slight abuse of notation, we assume $c_i(v_0,v_0)=0$.}
\end{lemma}
\begin{proof}
Observe that for every $H\in \comp(G^*_i)$ not containing vertex $v_0$, at least one red edge $(v_0,v)$, for some $v \in V(H)$, has to be contained in any MST of $(V(T_i),E(T_i)\cup F^*_i)$. Thus, for some $v \in V(H)$, at least one edge $(u,v)$ where $u$ is the parent of $v$ in $T$ has to be contained in any MST of $G=(V(T),E(T)\cup F^*)$. As $c_i(v_0,v)=c(u,v)$, the claim follows by summing over all components $H \in \comp(G^*_i)$ for all $i$'s.
\qed
\end{proof}

The key idea of our algorithm is to find a set $F$ of blue edges whose overall activation cost does not exceed the budget, and such that $(V,F)$ is a forest of stars. More precisely, for every $i=1,\dots,h$, the algorithm first finds a set $\hat F_i \subseteq B_i$ such that $\sum_{e \in \hat F_i}\gamma_i(e)\leq \Delta$ and $\hat G_i:=(V(T_i),\hat F_i)$ is a forest of stars; then, it considers the set $F_i:=\{e_{j,q}\mid \bar e_{j,q} \in \hat F_i\}$ of the corresponding blue edges for the original instance. Observe that (i) $G_i:=(V(T),F_i)$ is still a forest of stars and (ii) the overall activation cost of the edges in $F_i$ equals that of the edges in $\hat F_i$. Furthermore, using Equation~(\ref{eq: optimal pricing for set of leader's edges}), we can derive the following lemma, which claims that when we map $\hat{F}_i$ back to $F_i$ the obtained revenue cannot decrease:

\begin{lemma}\label{lm: converted solution}
$r(F_i)\geq r(\hat F_i)$.
\end{lemma}

We now give a lower bound of the revenue that can be obtained from $\hat{F}_i$. The bound trivially follows from (\ref{eq: optimal pricing for set of leader's edges}):

\begin{lemma}\label{lm: revenue of a blue star}
Let $L_i:=\{v \mid v \in V(T_i), \text{$v$ is a leaf of some star in $\hat G_i$}\}$.\footnote{If a star contains only one edge, then let any of its vertices to be a leaf.} Then, $r(\hat F_i)\geq \sum_{v\in L_i}c_i(v_0,v)$.
\end{lemma}

Next lemma essentially shows that there exists a solution for $T_i$ which is a forest of stars yielding a revenue of at lest a half of the optimal revenue for $T_i$.

\begin{lemma}\label{lm: decomposition of cycle trees}
Let $B' \subseteq B_i$ and let $U=\{v \mid\text{$v$ is an endvertex of some edge in $B'$}\}$. There exists a polynomial time algorithm that finds two sets $F^1$ and $F^2$ such that (i) $F^1,F^2\subseteq B'$,
(ii) both $(V(T_i),F^1)$ and $(V(T_i),F^2)$ are forests of stars, and (iii) $r(F^1)+r(F^2)\geq \sum_{v \in U}c_i(v_0,v)$.
\end{lemma}
\begin{proof}
Let $D$ be the graph induced by edge set $B'$. Let $D^j$ be any of the $t$ connected components in $D$, and let $T^j$ be any spanning tree in $D^j$. As $T^j$ is a bipartite graph, it is possible to partition the set of its vertices into two sets $V^j_1$ and $V^j_2$ in polynomial time. Moreover, by the connectivity of $T^j$, every vertex $v \in V^j_\ell$ ($\ell\in\{1,2\}$) is adjacent to some vertex in $V^j_{3-\ell}$, and thus it is easy to find a set $E^j_\ell$ of edges in $T^j$ such that $(V(D^j),E^j_\ell)$ is a forest of stars with centers in $V^j_\ell$ and leaves in $V^j_{3-\ell}$. Therefore, for $\ell = 1,2$, $F^\ell = \bigcup_{j=1}^t E^j_\ell$ are two sets of edges satisfying (i) and (ii). Furthermore, $\bigcup_{j=1}^t (V^j_1\cup V^j_2)=\bigcup_{j=1}^t V(D^j)=V(D)$. As a consequence, from Lemma~\ref{lm: revenue of a blue star}, (iii) is also satisfied.
\qed
\end{proof}

To compute $\hat F_i$, the algorithm does the following. Our algorithm uses the well-known $\fptas$ for the \emph{Knapsack Problem} to compute a $(1+\epsilon/(2h))$-approximate solution $S_i$ for the following instance of knapsack. For each $v_j$, consider the blue edge $e \in B_i$ incident to $v_j$ with cheapest activation cost. We create an object $o^i_j$ of profit $c_i(v_0,v_j)$ and volume $\gamma_i(e)$; we say that $e$ is associated with the object $o^i_j$. Finally, the volume of the knapsack is $\Delta$. Denote by $K_i$ the input instance of knapsack. The solution $S_i$ for $K_i$ identifies a set of blue edges, namely $B'=\{e \in B_i \mid \exists o^i_j \in S_i \mbox{\ s.t.\ } e \text{ is associated with } o^i_j\}$. The algorithm then uses the decomposition algorithm described in Lemma~\ref{lm: decomposition of cycle trees} to find two subsets of edges $F^1$ and $F^2$, and then it sets $\hat F_i$ to $F^1$ if $r(F^1)\geq r(F^2)$, and to $F^2$ otherwise.
The pseudocode of the algorithm is given in Algorithm~\ref{alg: knapsack like algorithm}.

\floatname{algorithm}{Algorithm}
\begin{algorithm}
\begin{algorithmic}[1]

\FOR{$i=1$ to $h$}

\STATE compute a $(1+\epsilon/(2h))$-approximate solution $S_i$ for the knapsack instance $K_i$

\STATE $B'=\{e \in B_i \mid o^i_j \in S_i, e \text{ is associated with } o^i_j\}$

\STATE compute $F^1$ and $F^2$ w.r.t. $B'$ as explained in Lemma~\ref{lm: decomposition of cycle trees}

\STATE{{\bf if} $r(F^1)\geq r(F^2)$ {\bf then } $\hat F_i:=F^1$ {\bf else } $\hat F_i:=F^2$ {\bf end if}}

\STATE $F_i:=\{e_{j,q}\mid \bar e_{j,q} \in \hat F_i\}$

\ENDFOR

\STATE {\bf return} the best of the $F_i$'s

\end{algorithmic}
\caption{}\label{alg: knapsack like algorithm}
\end{algorithm}

\begin{theorem}
Algorithm~\ref{alg: knapsack like algorithm} computes a $(2h+\epsilon)$-approximate solution in polynomial time for \stack$(\gamma,\Delta)$, for any constant $\epsilon >0$.
\end{theorem}
\begin{proof}
Remind that $G^*_i=( \{ v_0, \dots, v_{\ell_i} \}, F^*_i)$, where $F^*_i$ are obtained by mapping the edges of an optimal solution $F^*$ to the blue edges of the auxiliary instance $T_i$. In order to show a lower bound for the profit of the optimal solution of $K_i$, we define a feasible solution $S^*_i$ as follows: for each connected component $H$ of $G^*_i$, let $v\in V(H)$ be the vertex that minimizes $c(v_0, v)$, and consider any spanning tree $\mathcal{T}_H$ of $H$ rooted at $v$.
 Notice that for each $v_j \in V(H) \setminus \{v\}$ we have an object $o^i_j$ whose volume is at most $\gamma(\bar{v}_j, v_j)$, where $\bar{v}_j$ denotes the parent of $v_j$ in $\mathcal{T}_H$. Add such objects to the solution $S^*_i$.
	Once we have considered all the connected components of $G^*_i$, the solution $S^*_i$ contains a set of objects of total volume at most $\Delta$  (since the solution $F^*$ is feasible). Moreover, by construction, the profit of $S^*_i$ must be at least
\[
	{\tt profit}(S^*_i) \ge c(T_i) - \sum_{H \in \comp(G^*_i)} \min_{v \in V(H)} c_i(v_0, v) \mbox{.}
\]

We now bound the revenue $r(\hat{F}_i)$. Since, from Lemma \ref{lm: converted solution} and Lemma \ref{lm: decomposition of cycle trees}, $ r(F_i) \ge r(\hat{F}_i) \ge \frac{1}{2} \, \texttt{profit}(S_i)$, we have that
\begin{multline*}
	c(T_i) - \sum_{H \in \comp(G^*_i)} \min_{v \in V(H)} c_i(v_0, v) \le {\tt profit}(S^*_i)
	\\ \le \left( 1+\frac{\epsilon}{2h} \right) \texttt{profit}(S_i)
	\le \left( 2+\frac{\epsilon}{h} \right) r(\hat{F}_i) \le \left( 2+\frac{\epsilon}{h} \right) r(F_i) \mbox{.}
\end{multline*}

\noindent
By summing over all levels and using Lemma \ref{lm: upper bound on the leader's revenue}, we obtain
\begin{multline*}
	r(F^*) \le c(T) - \sum_{i=1}^{h} \sum_{H \in \comp(G^*_i)} \min_{v \in V(H)} c_i(v_0, v) \\
	\le \left( 2+\frac{\epsilon}{h} \right) \sum_{i=1}^{h} r(F_i)
	\le \left( 2h+\epsilon \right) \max_{i=1,\dots,h} r(F_i) \mbox{.}
\end{multline*}

\noindent
This completes the proof.
\qed
\end{proof}


\section{Open problems}
In this paper we have presented a collection of results concerning
some interesting variants of the classic $\stack$ game. Many
intriguing problems are left open. Among the others, we list the
following: (i) Is \stack$(0,0)$ \np-hard?
(ii) Can we design a better approximation algorithm for \stack$(0,0)$? (iii) Can we prove a stronger
inapproximability result for \stack$(\gamma,\Delta)$ than the one
holding for \stack?  (iv) What can we say about
\stack$(\gamma,\Delta)$ for instances with uniform activation costs? Is the problem \np-hard? Can we design and 
can we extend our results for \stack$(0,0)$? (v) Finally, and most importantly, does
$\stack$ admit a constant factor approximation algorithm?

\end{document}